\documentclass[aps,superscriptaddress,11pt,twoside,notitlepage]{revtex4-1}

\usepackage{graphicx,epic,eepic,epsfig,amsmath,latexsym,amssymb,verbatim,revsymb,color}
\usepackage[noautoscale,vcentermath]{youngtab}
\usepackage{young}

\newcounter{corollary3}

\usepackage{theorem}
\newtheorem{definition}{Definition}
\newtheorem{proposition}[definition]{Proposition}
\newtheorem{lemma}[definition]{Lemma}

\newtheorem{theorem}[definition]{Theorem}
\newtheorem{corollary}[definition]{Corollary}

\def\squareforqed{\hbox{\rlap{$\sqcap$}$\sqcup$}}
\def\qed{\ifmmode\squareforqed\else{\unskip\nobreak\hfil
\penalty50\hskip1em\null\nobreak\hfil\squareforqed
\parfillskip=0pt\finalhyphendemerits=0\endgraf}\fi}
\def\endenv{\ifmmode\;\else{\unskip\nobreak\hfil
\penalty50\hskip1em\null\nobreak\hfil\;
\parfillskip=0pt\finalhyphendemerits=0\endgraf}\fi}
\newenvironment{proof}{\noindent \textbf{{Proof~} }}{\qed}

\newlength{\blank}
\mathchardef\ordinarycolon\mathcode`\:
\mathcode`\:=\string"8000
\def\vcentcolon{\mathrel{\mathop\ordinarycolon}}
\begingroup \catcode`\:=\active
  \lowercase{\endgroup
  \let :\vcentcolon
  }

\newcommand{\nc}{\newcommand}
\nc{\rnc}{\renewcommand}
\nc{\beq}{\begin{equation}}
\nc{\eeq}{{\end{equation}}}
\nc{\beqa}{\begin{eqnarray}}
\nc{\eeqa}{\end{eqnarray}}
\nc{\lbar}[1]{\overline{#1}}
\nc{\bra}[1]{\langle#1|}
\nc{\ket}[1]{|#1\rangle}
\nc{\ketbra}[2]{|#1\rangle\!\langle#2|}
\nc{\braket}[2]{\langle#1|#2\rangle}
\nc{\proj}[1]{| #1\rangle\!\langle #1 |}
\nc{\avg}[1]{\langle#1\rangle}
\nc{\Rank}{\operatorname{rank}\,}
\nc{\smfrac}[2]{\mbox{$\frac{#1}{#2}$}}
\nc{\tr}{\operatorname{Tr}}
\nc{\ox}{\otimes}
\nc{\dg}{\dagger}
\nc{\dn}{\downarrow}
\nc{\cA}{{\cal A}}
\nc{\cB}{{\cal B}}
\nc{\cC}{{\cal C}}
\nc{\cD}{{\cal D}}
\nc{\cE}{{\cal E}}
\nc{\cF}{{\cal F}}
\nc{\cG}{{\cal G}}
\nc{\cH}{{\cal H}}
\nc{\cI}{{\cal I}}
\nc{\cJ}{{\cal J}}
\nc{\cK}{{\cal K}}
\nc{\cL}{{\cal L}}
\nc{\cM}{{\cal M}}
\nc{\cN}{{\cal N}}
\nc{\cO}{{\cal O}}
\nc{\cP}{{\cal P}}
\nc{\cR}{{\cal R}}
\nc{\cS}{{\cal S}}
\nc{\cT}{{\cal T}}
\nc{\cX}{{\cal X}}
\nc{\cZ}{{\cal Z}}
\nc{\csupp}{{\operatorname{csupp}}}
\nc{\qsupp}{{\operatorname{qsupp}}}
\nc{\var}{\operatorname{var}}
\nc{\rar}{\rightarrow}
\nc{\lrar}{\longrightarrow}
\nc{\polylog}{\operatorname{polylog}}
\nc{\1}{{\openone}}

\nc{\RR}{{{\mathbb R}}}
\nc{\CC}{{{\mathbb C}}}
\nc{\FF}{{{\mathbb F}}}
\nc{\NN}{{{\mathbb N}}}
\nc{\ZZ}{{{\mathbb Z}}}
\nc{\PP}{{{\mathbb P}}}
\nc{\QQ}{{{\mathbb Q}}}
\nc{\UU}{{{\mathbb U}}}
\nc{\EE}{{{\mathbb E}}}
\nc{\id}{{\operatorname{id}}}

\nc{\be}{\begin{equation}}
\nc{\ee}{{\end{equation}}}
\nc{\bea}{\begin{eqnarray}}
\nc{\eea}{\end{eqnarray}}
\nc{\<}{\langle}
\rnc{\>}{\rangle}
\nc{\Hom}[2]{\mbox{Hom}(\CC^{#1},\CC^{#2})}
\nc{\rU}{\mbox{U}}

\nc{\ob}[1]{#1}

\nc{\LO}{\text{LO}}
\nc{\LOCC}{\text{LOCC}}
\nc{\cLOCC}{{\overline{\text{LOCC}}}}
\nc{\SEP}{\text{SEP}}
\nc{\PPT}{\text{PPT}}

\newcommand*{\anti}{{\,\yng(1,1)}}
\newcommand*{\sy}{{\,\yng(2)}}
\newcommand*{\complex}{\mathbf{C}}

\newcommand*{\half}{\frac{1}{2}}

\newcommand*{\singlebox}{{\,\yng(1)}}
\newcommand*{\col}{{\,\yng(1,1,1,1)}}
\newcommand*{\boxx}{{\,\yng(2,2)}}
\newcommand*{\hook}{{\,\yng(2,1,1)}}
\newcommand*{\sym}{\mathrm{Sym}}
\newcommand*{\sign}{\operatorname{sign}\,}
\newcommand*{\dcol}{{  \textrm{{\scriptsize d}} \bigg\{ \, \yng(1,1,1,1,1,1,1) }}
\newcommand*{\dhook}{{\textrm{{\scriptsize d-1}} \Big\{ \,\yng(2,1,1,1,1,1)}}
\newcommand*{\dboxx}{{ \textrm{{\scriptsize d-2}} \Big\{\,\yng(2,2,1,1,1)}}

\newcommand*{\dualcol}{{  \textrm{{\scriptsize d-2}} \Big\{ \, \yng(1,1,1,1,1) }}

\begin{document}

\title{Entanglement of the Antisymmetric State}

\author{Matthias Christandl}
\affiliation{Institute for Theoretical Physics, ETH Zurich, Wolfgang-Pauli-Strasse 27, 8093 Zurich, Switzerland}
\email{christandl@phys.ethz.ch}

\author{Norbert Schuch}
\affiliation{Institute for Quantum Information, California Institute of
Technology, Pasadena CA 91125, U.S.A.}
\affiliation{Max-Planck-Institut f\"ur Quantenoptik, Hans-Kopfermann-Str.~1, D-85748 Garching, Germany}
\email{norbert.schuch@googlemail.com}

\author{Andreas Winter}
\affiliation{Department of Mathematics, University of Bristol, Bristol BS8 1TW, U.K.}
\affiliation{Centre for Quantum Technologies, National University of Singapore, 2 Science Drive 3, Singapore 117542}
\email{a.j.winter@bris.ac.uk}


\begin{abstract}
We analyse the entanglement of the antisymmetric state in dimension $d \times d$ and present two main results. First, we show that the amount of secrecy that can be extracted from the state is low, more precisely, the distillable key is bounded by $O(\frac{1}{d})$. Second, we show that the state is highly entangled in the sense that a large number of ebits are needed in order to create the state: entanglement cost is larger than a constant, independent of $d$. The second result is shown to imply that the regularised relative entropy with respect to separable states is also lower bounded by a constant. 
%
Finally, we note that the regularised relative entropy of entanglement is asymptotically 
continuous in the state.

Elementary and advanced facts from the representation theory of the unitary group,
including the concept of plethysm, play a central role in the proofs of the
main results.
\end{abstract}

\maketitle

\Yboxdim{4pt}

\section{Introduction}
Entanglement is a quantum phenomenon governing the correlations between two quantum systems. It is both responsible for Einstein's ``spooky action at a distance''~\cite{Einstein1971} as well as the security of quantum key distribution~\cite{BB84,E91}. Quantum key distribution, or QKD for short, is a procedure to distribute a perfectly secure key among two distant parties, something that is not possible in classical cryptography without assumptions on the eavesdropper. 

In the early days of quantum information theory, it was quickly realised that the universal resource for bipartite entanglement is the ebit, that is, the state $\ket{\psi}:=\frac{1}{\sqrt{2}} (\ket{00}+\ket{11})$~\cite{BBPS96}. Ebits are needed for teleportation~\cite{teleportation}, superdense coding~\cite{superdensecoding} and directly lead to secret bits~\cite{E91, BBM92}. It is therefore natural to associate the usefulness of a quantum state with the amount of ebits that can be extracted from it or the amount of ebits needed to create the state~\cite{BDSW96}. Formally, one considers the \emph{distillable entanglement}
\begin{equation}
  \label{eq:E_D}
  E_D(\rho)=\lim_{\epsilon \rightarrow 0} \lim_{n \rightarrow \infty} \sup_{\Lambda_n \in \text{ LOCC} } 
  \left\{ \frac{m}{n}: \|\Lambda_n(\rho^{\otimes n})-\psi^{\otimes m} \|_1 \leq \epsilon \right\},
\end{equation}
and the \emph{entanglement cost}
\begin{equation}
  \label{eq:E_C}
  E_C(\rho)=\lim_{\epsilon \rightarrow 0} \lim_{n \rightarrow \infty} \inf_{\Lambda_n \in \text{ LOCC} } 
  \left\{ \frac{m}{n}: \|\Lambda_n(\psi^{\otimes m})-\rho^{\otimes n} \|_1 \leq \epsilon \right\},
\end{equation}
where the supremum and infimum ranges over all completely positive trace
preserving (CPTP) maps that can be obtained from local operations and
classical communication (LOCC) on the state (this is, operations which can
be implemented using a multi-round protocol where in every round, both
parties carry out some local operation, followed by an
exchange of classical information~\cite{nielsen-chuang}).
For ease of notation we write $\psi$ short for $\proj{\psi}$.

An important result relating to these quantities has been the discovery of bound entanglement, that is of states that need ebits for their creation but from which no ebits can be extracted asymptotically: $E_C(\rho)>0$ and $E_D(\rho)=0$~\cite{bound-e-paper}. A recent surprise has been the realization that there exist bound entangled states from which secrecy can be extracted~\cite{horodecki-2005-94}, a result that overthrew previous beliefs that secrecy extraction and entanglement distillation would go hand in hand.

This has motivated research into the amount of key that can be distilled from a quantum state as an entity in its own right. The \emph{distillable key} is defined as
\begin{equation}
  K_D(\rho_{AB}) = \lim_{\epsilon \rightarrow 0} \lim_{n \rightarrow \infty} \sup_{\Lambda_n \text{ LOCC}, \gamma_m } \left\{ \frac{m}{n}: \| \Lambda_n(\rho^{\otimes n})-\gamma_m \|_1 \leq \epsilon \right\},
  \label{eq:K_D}
\end{equation}
where $\gamma_m$ denotes a quantum state which contains $m$ bits of pure secrecy 
(see Definition~\ref{def:secrecy}).

A fundamental question at this point is this: Do there exist states which
require key to create them but from which no key can be distilled? 
Note that a mathematical formulation of this question appears to require
the definition of a ``key cost'' of a state, which is problematic
since the states $\gamma_m$ containing $m$ bits of 
pure key, the \emph{private states} of Definition~\ref{def:secrecy},
form a heterogenous class of states which are not all equivalent to
each other.
Even the weaker form of this question, whether there exist states with
$E_C(\rho) > 0$ but $K_D(\rho) = 0$, seems too difficult at the moment,
since we have apart from the separability of $\rho$ no criterion for $K_D(\rho) = 0$.
Here we show that in an asymptotic sense the answer is yes: in the spirit
of~\cite{RennerWolf:boundkey}, we show that there exists a family of
states with constant lower bound on their entanglement cost, but
arbitrarily small distillable key. These results have been 
previously reported in~\cite{ChristandlHighly2009}.

In order to derive this result, we make use of the theory of entanglement
with its many entanglement measures. The motivation for this is the
following. Due to the asymptotic nature of the definitions it is a difficult task to
evaluate the distillable entanglement, the entanglement cost and the
distillable key on specific quantum states. All three quantities have in common that they
measure the amount of entanglement in a quantum state, i.e. they do not
increase under LOCC operations, they
vanish on separable states (i.e., states which can be written as 
a convex combination of product states, $\rho=\sum_i p_i
\rho_i^A\otimes \rho_i^B$),
and they equal one when evaluated on an ebit.
This has led to an axiomatisation of the quantities that measure
entanglement and to the definition of a whole zoo of entanglement
measures (cf.~\cite{christandlPhD}). One of the main uses of all the new entanglement measures is
that they are mostly sandwiched between distillable entanglement (or even
distillable key) and entanglement cost and hence form upper and lower
bounds for these quantities. Even though these new entanglement measures
often involve complicated minimisations or asymptotic limits they are
sometimes easier to calculate than distillable entanglement, distillable
key and entanglement cost. 

The states for which entanglement measures have been calculated are typically 
characterised by their symmetry. The most prominent example are so-called 
Werner states in dimension $d\times d$, defined by the
property
\[
  (g\otimes g) \rho (g^\dagger \otimes g^\dagger) = \rho
\]
for all $g\in U(d)$, the unitary group. Werner states can be given explicitly as the one parameter family
\[
  \rho= p \sigma_d+(1-p)\alpha_d,
\]
where $p\in [0,1]$.
Here, $\sigma_d$ is the state proportional to the projectors onto the 
symmetric subspace and $\alpha_d$ is the state proportional to the projector 
onto the antisymmetric subspace. In this work we will bound the value of 
certain entanglement measures for the totally antisymmetric states $\alpha_d$.

The first entanglement measure we use is the squashed entanglement~\cite{squashed},
\begin{equation}
  E_{sq}(\rho_{AB})=\inf_{\rho_{ABE}: \rho_{AB}=\tr_E \rho_{ABE}} \half I(A;B|E)_{\rho},
  \label{eq:E_sq}
\end{equation}
where $I(A;B|E)_\rho=H(AE)_\rho+H(BE)_\rho-H(ABE)_\rho-H(E)_\rho$ is the quantum conditional mutual information,
with $H(X)_\rho=H(\rho_X)$ the von Neumann entropy of the reduced
state on $X$, $H(\sigma):=-\tr\sigma\log_2\sigma$.
We show that squashed entanglement is an upper bound on the distillable key and hence establish the chain of inequalities
\begin{equation}
  E_D \leq K_D \leq E_{sq} \leq E_C.
  \label{eq:chain}
\end{equation}
A concrete calculation of a bound on the squashed entanglement of the antisymmetric states will yield our first main result, an upper bound on the distillable key. 
\begin{theorem}\label{th:upper}
\begin{align}
  \label{eq:KD-upper}
  K_D(\alpha_d) &\leq \left. \begin{cases}
                               \log_2\frac{d+2}{d} & \text{ if } d \text{ is even}    \\
                               \half\log_2\frac{d+3}{d-1} & \text{ if } d \text{ is odd}
                             \end{cases}
                       \right\} = O\left(\frac{1}{d}\right).
\end{align}
\end{theorem}

In order to find a lower bound on the entanglement cost of the antisymmetric state, we will use its charaterisation as the regularised entanglement of formation $E_C=E_F^\infty$. The entanglement of formation is defined as 
\begin{equation}
  \label{eq:E_F}
  E_F(\rho) = \min_{\{p_i, \proj{\varphi}_i \}_i: \rho = \sum_i p_i \proj{\varphi_i} }\sum_i p_i H\bigl(\tr_B \proj{\varphi_i}\bigr),
\end{equation}
and its regularisation is given by
\begin{align} \label{regEoF} 
        E_F^\infty(\rho):=\lim_{n\rightarrow\infty} \frac{1}{n} E_F\bigl( \rho^{\ox n} \bigr). \end{align}
Making heavy use of the symmetry of the antisymmetric state we will relax the minimisation in the definition of the entanglement of formation to a linear programme and obtain the second main result of this paper.
\begin{theorem}  \label{th:EC-lower}
$\displaystyle{E_C(\alpha_d) \geq \log_2 \frac{4}{3} \approx 0.415.}$
\end{theorem}
It is not difficult to see that the entanglement of formation of $\alpha_d$ equals one and hence that the truth of the additivity conjecture for entanglement of formation would have implied $E_C(\alpha_d)=1$. Since Hastings has provided a counterexample~\cite{Hastings} to the additivity conjecture~\cite{Shor:equivalences}, this consequence is put into doubt and the only evidence for $E_C(\alpha_d)=1$ was Yura's brute force calculation which proved this statement for $d=3$. Our result can therefore be seen as supporting evidence for $E_C(\alpha_d)=1$, and at least provides a further example where some weak form of additivity holds. At present the techniques in this paper are not sufficient to prove $E_C(\alpha_d)=1$, but further development may be capable of doing so.

Using the tools developed to prove Theorem~\ref{th:EC-lower}, we obtain a
lower bound to the regularised relative entropy of entanglement with
respect to separable states.
\begin{corollary}
        \label{cor:ERinf_lowerbnd}
$\displaystyle{E_{R,\mathrm{sep}}^\infty(\alpha_d) \geq \log_2 \sqrt{
\frac{4}{3}} \approx 0.2075.}$ 
\end{corollary}
Here, the relative entropy of entanglement (with respect to separable states) is defined as
\begin{align*}
E_{R,\mathrm{sep}}(\rho):=\min_{\sigma \  \text{separable}} D(\rho||\sigma),
\end{align*}
where $D(\rho||\sigma):=\tr \rho [\log \rho -\log \sigma]$,
and the regularised relative entropy of entanglement is
\begin{align} \label{def:relentasympt}
E_{R,\mathrm{sep}}^\infty(\rho)=\lim_{n\rightarrow\infty}\tfrac{1}{n}
E_{R,\mathrm{sep}}(\rho^{\otimes n}).
\end{align}
From the point
of view of entanglement theory, this result is interesting for at least
three reasons. First, it shows that the additivity violation of the
relative entropy of entanglement for the antisymmetric state, first
observed in~\cite{VollbrechtWerner01}, is not very strong in the
asymptotic limit. Secondly, the regularised relative entropy of
entanglement with respect to separable states behaves very differently from
the relative entropy of entanglement with respect to PPT states, as the
latter takes the value $\log_2 \frac{d+2}{d}$ on $\alpha_d$~\cite{audenaert-2001-87}.
Thirdly, it shows that the relative entropy of entanglement can sometimes
be larger and sometimes be smaller than the squashed entanglement.
Finally, we note that as an entanglement measure, the relative entropy of
entanglement with respect to separable states
satisfies~\cite{horodecki-2005-94} 
\[
  E_D \leq K_D \leq E_{R,\mathrm{sep}}^\infty \leq E_C. 
\]
and that it is asymptotically continuous, as we show in Proposition~\ref{prop-rel-ent-cont}. 

In order to derive both main results of the paper we make use of the symmetry 
properties of the antisymmetric state and the associated representation theory 
of the unitary group in dimension $d$~\cite{FultonHarris91}. 
For the lower bound on entanglement cost, we relax the calculation of 
$E_F(\alpha_d^{\otimes n})$ in a first step into a semidefinite programme which 
we reduce in a second step with the help of representation theory (for the 
first time using the concept of a plethysm in quantum information theory) into 
a linear programme~\cite{LP-book}. We then find a feasible point of the dual 
for the latter, which results in our lower bound of $\log_2 \frac{4}{3}$ 
for entanglement cost. On the way we recover Yura's result for $d=3$.

The rest of the paper is organised as follows. In Section~\ref{sec:prelim} we introduce the notation from representation theory that will be used throughout the paper. In Section~\ref{sec:key} we prove the upper bound on the squashed entanglement and distillable key. 
In Section~\ref{sec:cost} we exhibit the sequence of relaxations that will lead to the lower bound on the entanglement cost. In Section~\ref{sec:relent} we will derive the lower bound on regularised relative entropy of entanglement of the antisymmetric state with respect to separable states. Furthermore, we establish that it is asymptotically continuous as a function of the state. We will conclude the paper with remarks and open questions in Section~\ref{sec:conclusion}. The appendices contain details on the representation-theoretic calculations and the linear programme.

\section{Representation theoretic preliminaries}\label{sec:prelim}

Representations of the unitary group $U(d)$ can be taken to be unitary and 
decompose into a direct sum of irreducible representations. The latter are 
classified according to their highest weight. For each dominant weight 
$\lambda $, i.e. $\lambda=(\lambda_1, \ldots, \lambda_d)$ with 
$\lambda_i \geq \lambda_{i+1} \in {\bf Z}$ there is exactly one irreducible 
representation $V_\lambda$. When $\lambda_d \geq 0$, we write 
$\lambda \vdash_d n$ if $n:=|\lambda |:=\sum_i \lambda_i$. 
Such $V_\lambda$ can be viewed as a subrepresentations of the $n$-fold 
diagonal action of the unitary group on $(\complex^d)^{\otimes n}$:
$$T^{n}: g \mapsto g^{\otimes n},$$
since by Schur-Weyl duality
$$T^{n}\cong \bigoplus_{\lambda \vdash_d n} V_\lambda \otimes \complex^{\dim [\lambda]},$$
where $[\lambda]$ denotes the $S_n$-Specht module corresponding to the Young frame $\lambda$. In the following we will often use the interpretation of $\lambda$ as a Young frame, i.e. as a diagram of boxes arranged in $d$ rows with $\lambda_i$ boxes in row $i$, and use the corresponding diagrammatic notation. As a vector space, $V_\lambda$ can be constructed as the image of the Young symmetriser, a certain element in the group algebra of $S_n$, when applied to $(\complex^d)^{\otimes n}$. 
The projector onto $V_\lambda$ is denoted by $P_\lambda$.

Two types of representations are of particular importance. First, the
symmetric representations with Young diagram $\lambda=(n, 0, \ldots, 0)$
which act on the totally symmetric subspace $\sym^n(\complex^d) $ of
$(\complex^d)^{\otimes n}$. Second, the fundamental representations with
Young diagram $\lambda=(1, 1, \ldots, 1, 0, 0, \ldots, 0)$ which act on
the totally antisymmetric subspace $\wedge^n(\complex^d)$ of
$(\complex^d)^{\otimes n}$. Note that the latter are zero-dimensional for
$d<n$. 

The dimension of $V_\lambda$ is given by Weyl's dimension formula 
\begin{align} \label{eq:Weyldim}
  \dim V_\lambda=\frac{\prod_{i<j} (\lambda_i-\lambda_j-i+j)}{\prod_{k=1}^{d-1} k!}
\end{align}
and specializes in the case of a fundamental representation to $\binom{d}{n}$.

The first case of interest to us is $n=2$, where
$$T^2 \cong V_{(1,1)} \oplus V_{(2,0)},$$ 
or in diagrammatic notation 
$$\singlebox^{\otimes 2}\cong \anti \oplus \sy\;.$$
It then follows immediately from Schur's lemma that the $U(d)$-invariant states on this space must be of the form
$$\rho=p \sigma_d +(1-p) \alpha_d,$$
where $p\in [0, 1]$ and with the totally antisymmetric and totally symmetric states
$$\sigma_d= \frac{2}{d(d+1)} P_\sy\;,$$
$$\alpha_d= \frac{2}{d(d-1)} P_\anti\;,$$
respectively. Note that we suppress the dependence on $d$ when the dimension is clear from the context.
Later we will compute similar decompositions of more complicated type.

\section{Upper bound on the distillable key}
\label{sec:key}
In this section we will first show that squashed entanglement is an upper bound to the amount of key that one can distill from quantum states. Then we will find an upper bound on squashed entanglement of the antisymmetric state. Together, this proves Theorem~\ref{th:upper}.

Recall the definition of squashed entanglement and the definition of the key rate. The latter contains a maximisation over private states that contain $m$ bits of pure secrecy, the formal definition of such states follows.

\begin{definition}[\cite{horodecki-2005-94}]\label{def:secrecy}
A private state containing $m$ bits of secrecy is a state $\gamma_m$ of the form
$$\gamma_m=U \sigma_{AA'BB'}U^\dagger$$ for some unitary
$U=\sum_i \proj{ii} \otimes U_i$ and
$\sigma_{AA'BB'}=\Phi_{AB} \otimes \sigma_{A'B'}$, where
$\ket{\Phi}=\frac{1}{\sqrt{2^m}} \sum_{i=1}^{2^m} \ket{i}\ket{i}$ is
the maximally entangled state of rank $2^m$. 
System $AB$ is known as the \emph{key part} of the state and system 
$A'B'$ is known as the \emph{shield part}.
\end{definition}

\begin{lemma}[\cite{christandlPhD}]
  \label{lemma:keysquashed} 
  For all bipartite quantum states $\rho_{AB}$,
  $$K_D(\rho_{AB})\leq E_{sq}(\rho_{AB}).$$
\end{lemma}
\begin{proof}
Let $\Lambda_n$ be a CPTP map that can be implemented with an LOCC protocol and that satisfies
$$\|\Lambda_n(\rho^{\otimes n})- \gamma_m \|_1 \leq \epsilon,$$
and assume that the dimension of the $A'B'$ part is at most exponential in $n$. This last assumption can be made without loss of generality since the optimal key distillation protocol can be approximated by a sequence of protocols satisfying this requirement. In order to see this, note that one can stop the optimal protocol when the extracted bits are almost perfect and use privacy amplification~\cite{RenKoe05} to make them perfect. 
The communication needed in order to achieve privacy amplification amounts to the choice of a function from a set of two-universal hash functions. Classes of such functions of size exponential in $n$ exist~\cite{CarWeg79}.
This shows that privacy amplification needs an amount of communication that is at most linear in the amount of bits extracted. Therefore, without loss of generality, the dimension the shield size can be assumed to grow at most exponentially in $n$, say $\leq c^n$ for some $c\geq 1$. 

Since squashed entanglement is
a monotone under LOCC~\cite{squashed} and asymptotically continuous~\cite{AliFan04}
$$E_{sq}(\rho^{\otimes n}) \geq E_{sq}(\Lambda_n(\rho^{\otimes n})) \geq E_{sq}(\gamma_m)-16 c \sqrt{\epsilon} n \log_2 d-4 h(2\sqrt{\epsilon}).$$
Recall from Definition~\ref{def:secrecy} the form of the state $\gamma_m\equiv \gamma_{AA'BB'}$.
In order to show
that $E_{sq}(\gamma_m) \geq m$, consider an arbitrary extension
$\gamma_{AA'BB'E}$ of $\gamma_{AA'BB'}$, which induces an extension 
$\sigma_{AA'BB'E}=
(U^\dagger\otimes \openone_E) \gamma_{AA'BB'E} (U\otimes \openone_E)
=\Phi_{AB}\otimes \sigma_{A'B'E}$ of $\sigma_{AA'BB'}$ in
Definition~\ref{def:secrecy}.  Clearly,
$$ H(AA'BB'E)_\gamma=H(AA'BB'E)_\sigma= H(A'B'E)_\sigma=H(A'B'E)_{\sigma_i},$$
with $\sigma_i:= U_i \otimes \openone_E \sigma_{A'B'E}
U_i^\dagger \otimes \openone_E$. Since
furthermore
$H(E)_{\gamma}=H(E)_\sigma=H(E)_{\sigma_i}$, 
we have that 
\[
H(AA'BB'|E)_\gamma = H(A'B'|E)_{\sigma_i}\ .
\]
Also, since
$H(AA'E)_\gamma=m+\frac{1}{2^m}\sum_i H(A'E)_{\sigma_i}$,
it follows that
$$
H(AA'|E)_\gamma=m+\frac{1}{2^m}\sum_i H(A'|E)_{\sigma_i}\ ,
$$
and similarly for $H(BB'E)_\gamma$. Altogether this gives
\begin{align*}
I(AA';BB'|E)_\gamma 
&= 
H(AA'|E)_\gamma+H(BB'|E)_\gamma - H(AA'BB'|E)_\gamma\\
&= 2m+\frac{1}{2^m}   
        \sum_i I(A';B'|E)_{\sigma_i} \\
&\geq 2m,
\end{align*}
where the non-negativity of the quantum conditional mutual information was
used in the last inequality. This shows that $E_{sq}(\gamma_m)\geq m$
and therefore 
$$E_{sq}(\rho) \geq \frac{m}{n}-16 c\sqrt{\epsilon}
\log_2 d-\frac{4}{n} h(2\sqrt{\epsilon}),$$ 
with the right hand side of this inequality converging to $K_D(\rho_{AB})$.
\end{proof}

\medskip

The following lemma provides an upper bound on the squashed entanglement of the antisymmetric state.
\begin{lemma} \label{lemma:squashedupper}
For even $d$ we have
$$E_{sq}(\alpha_d) \leq \log_2 \frac{d+2}{d}.$$
For odd $d$,
$$E_{sq}(\alpha_d) \leq \half \log_2 \frac{d+3}{d-1}.$$
\end{lemma}

\begin{proof}
Let $P_k$ be the projector onto the $\wedge^k(\complex^d)$ in
$(\complex^d)^{\otimes k}$. Recall that $d_k:=\dim \wedge^k(\complex^d)=
\binom{d}{k}$ and define $\rho_{ABE}:=\frac{P_k}{d_k}$, where $\cH_A
\cong \cH_B \cong \complex^d$ correspond to the first and the second
tensor factor and $\cH_E \cong (\complex^d)^{\otimes k-2}$ to the last
$k-2$ factors. It is clear that the reduced density matrix $\rho_{AB}:=
\tr_E \rho_{ABE}$ equals the totally antisymmetric state $\alpha_d$, or
conversely, that $\rho_{ABE}$ is an extension of $\alpha_{d}$. For this
extension we evaluate the conditional mutual information:
$$I(A;B|E)_{\rho}=H(AE)_\rho+H(BE)_\rho-H(E)_\rho-H(ABE)_\rho=\log_2
\frac{d_{k-1}^2}{d_{k-2}d_k}=\log_2 \frac{k}{k-1}\frac{d-k+2}{d-k+1}.$$
Minimising this function over different values of $k \in \{2, \ldots, d\}$ we find that for even $d$ the minimum value $I(A;B|E)_{\rho}=2\log_2 \frac{d+2}{d}$ is reached when $k=\frac{d}{2}+1$ and for odd $d$ the minimum value $I(A;B|E)_{\rho}=\log_2 \frac{d+3}{d-1}$ is reached when $k=\frac{d+1}{2}$. 
\end{proof}

\medskip
It is surprising that the bound from Lemma~\ref{lemma:squashedupper} for even dimension coincides with values of other entanglement measures~\cite{audenaert-2001-87}:
$$E_{R,\mathrm{PPT}}^\infty(\alpha_d)=E_{\text{Rains}}(\alpha_d)=E_{N}(\alpha_d)=\log_2 \frac{d+2}{d},$$ 
where $E_{R, \mathrm{PPT}}^\infty$ is the regularised relative entropy of entanglement with 
respect to PPT states (a PPT state is a state whose partial transpose is a positive semidefinite operator), $E_{\text{Rains}}$ is the Rains bound and $E_N$ is the logarithmic negativity. In the light of these results we are tempted to conjecture that $E_{sq}(\alpha_d) = \log_2 \frac{d+2}{d}$, at least for even $d$.

With the upper bound on squashed entanglement we not only match the best known upper bounds on distillable entanglement (for even dimension) but obtain new bounds even on the distillable key, since Lemma~\ref{lemma:keysquashed} and Lemma~\ref{lemma:squashedupper} prove Theorem~\ref{th:upper}.

Note also that our bound gives $E_{sq}(\alpha_d) \leq \frac{ 2 \log_2
e}{d-1} = O(\frac{1}{d})$ which improves over the bound
$E_{sq}(\alpha_d) =O(\frac{\log_2 d}{d})$ that was obtained using
the monogamy of squashed entanglement~\cite{Aaronson-squashed}. Note
finally, that the best known lower bound for both $E_D$ and $K_D$ is given
by $\frac{1}{d}$. Up to a constant, the bound that we have obtained for
squashed entanglement, distillable key (and distillable entanglement, but
this we knew before) is therefore optimal. Previously the best known upper
bound for distillable key was one half and stems from a computation of the
relative entropy of entanglement with respect to separable states (for two
copies) of Vollbrecht and Werner who showed that $E_{R,
\mathrm{sep}}(\alpha_d^{\otimes 2})\leq 1- \log_2
\frac{d-1}{d}$~\cite{VollbrechtWerner01} and hence $E_{R,
\mathrm{sep}}^\infty(\alpha_d)\leq \half E_{R,
\mathrm{sep}}(\alpha_d^{\otimes 2})= \half +O(\frac{1}{d})$. The
latter is an upper bound on $K_D$~\cite{horodecki-2005-94}.

\section{Lower bound on the entanglement cost}
\label{sec:cost}
The calculation of the entanglement cost using its characterisation as the regularised entanglement of formation, equation~\eqref{regEoF}, seems very daunting in general due to the infinite limit; but in fact, even the computation of entanglement of formation according to eq.~(\ref{eq:E_F}) is a very difficult task. However, for the antisymmetric states
$\alpha_d$ (and many copies thereof), the $g\ox g$ symmetry (for unitary
$g$) comes to help:

\begin{lemma}
  \label{lemma:renyi-2}
  For all $d\geq 3$,
  \[
    E_F(\alpha_d^{\ox n}) \geq - \log_2 \max_{\ket{\psi}_{A^nB^n} \in \anti^{\otimes n}} \tr \psi_{A^n}^2,
  \]
  where $\psi_{A^n}=\tr_{B^n} \proj{\psi}_{A^nB^n}$. Consequently,
 \begin{align}\label{eq:yura}
    E_C(\alpha_d) \geq - \lim_{n\rightarrow\infty} 
                            \frac{1}{n} \log_2 \max_{\ket{\psi}_{A^nB^n} \in \anti^{\otimes n}} \tr \psi_{A^n}^2.
   \end{align}
\end{lemma}
The quantity $\tr \psi_{A^n}^2$ is also known as the \emph{purity} of
$\psi_{A^n}$; it equals one for pure states, and is strictly smaller than one 
if the state is mixed.

\begin{proof}
Recall the definition of entanglement of formation in the case of a tensor
product state 
$
E_F(\alpha_d^{\otimes n})=
    \min_{
        \{p_i, \ket{\psi_i}\}:
        \alpha_d^{\otimes n}=\sum_i p_i \proj{\psi_i}
    } 
    \sum_i p_i H(\psi_{A,i})
$
and note that all states appearing in the ensembles are contained in
$\anti^{\otimes n}$. Thus $E_F(\alpha_d^{\otimes n}) \geq
\min_{\ket{\psi}_{A^nB^n} \in \anti^{\otimes n}}  H(\psi_{A^n}) $. [This
is in fact an equality, as any minimizer $\ket{\psi}_{A^nB^n}$ yields an
optimal decomposition $\int \proj{\psi_{g_1\dots g_n}} dg_1\cdots
dg_{n}$ of $\alpha_d^{\otimes n}$, with $\ket{\psi_{g_1\dots g_n}} =
(g_1\otimes\cdots\otimes g_n)^{\otimes 2}\ket\psi_{A^nB^n}$, and
$dg_i$ the Haar measure on $U(d)$.] The proof
follows by noting that the von Neumann entropy is lower bounded by the
quantum collision entropy (or quantum R\'enyi entropy of order two)
$H_2(\sigma)=-\log_2 \tr \sigma^2$ and from the formula $E_C(\rho)=\lim_{n
\rightarrow \infty}  \frac{1}{n}E_F(\rho^{\otimes n})$.  \end{proof}

\medskip
Yura~\cite{Yura:E_C} has used this bound and shown that the right hand side of~\eqref{eq:yura} equals $1$
if $d=3$. Together with the observation that the $E_C(\rho)\leq E_F(\rho)
\leq \frac{2}{d(d-1)}\sum_{i<j} H(\psi_{A, ij}) = 1$, where
$\ket{\psi_{ij}}=\frac{1}{\sqrt{2}}(\ket{ij}-\ket{ji})$, he has thus
calculated the entanglement cost of the antisymmetric state in this case. In
the following, we will reproduce Yura's result for $d=3$ and furthermore
show that the right hand side of~\eqref{eq:yura} is lower bounded $\log_2\frac{4}{3} \gtrsim 0.415$
for all $d$. 

In order to do so, we will first employ representation theory of the unitary and symmetric group as well as a relaxation in order to reduce the problem to a linear programme. In a second step, we will put a lower
bound on the optimal value of this programme using linear programming duality.

\begin{lemma}
  \label{lemma:symmetrization}
  We have
  \begin{equation}
    \label{eq:Werner} 
    \max_{\ket{\psi}_{A^nB^n} \in \anti^{\otimes n}}  \tr \psi_{A^n}^2
             = \max \tr \Omega_{A^nB^nA'^nB'^n} (F_{A^n:A'^n} \otimes \1_{B^nB'^n}),
  \end{equation}
where $F_{C:D}$ is the operator that permutes (``flips'') systems $C$ and $D$, 
and  where the maximisation on the right hand side is over all states of the form
  \begin{equation}
    \label{eq:sep}
    \Omega_{A^nB^nA'^nB'^n} = \sum_{y^n \in \{ \col, \boxx, \hook\}^n} 
                                 p_{y_1\ldots y_n}\rho_{y_1}\otimes  \cdots \otimes \rho_{y_n}
  \end{equation}
  that are separable across the $A^nB^n : A'^nB'^n$ cut. 
  The $p_{y^n}$ form a probability distribution symmetric under interchange of the 
  variables. $p_{y^n}$ vanishes if the number of $\hook$'s is odd. The states $\rho_y$ are proportional to projectors
  onto orthogonal subspaces of $\anti^{\otimes 2}$ which are isomorphic to
  irreducible representations of $U(d)$ with Young diagrams $\col$, $\boxx$ and $\hook$
  -- see Lemma~\ref{lem:rep} in Appendix~\ref{app:A}.
\end{lemma}
\begin{proof}
Note that $\tr \psi_{A^n}^2 = \tr (\psi_{A^n}\otimes \psi_{A'^n}) F_{A^n:A'^n}$.
Since $A^n=A_1 \cdots A_n$ and likewise for $A'^n$, we have 
$F_{A^n:A'^n}=F_{A:A'}^{\otimes n}$ and therefore 
\[
  \tr \psi_{A^n}^2 = \tr (\psi_{A^nB^n}\otimes \psi_{A'^nB'^n}) (F_{A:A'}^{\otimes n} \otimes \1_{B^nB'^n}).
\]
Because $F_{A:A'}$ commutes with $g^{\otimes 2}$ for all unitary $g$, we can replace 
$\psi_{A^nB^n}\otimes \psi_{A'^nB'^n} $ by the twirled state
\[
  \Omega_{A^nB^nA'^nB'^n} = \cT_{ABA'B'}^{\otimes n} (\psi_{A^nB^n}\otimes \psi_{A'^nB'^n}),
\]
where $\cT_{ABA'B'}$ is the twirling (CPTP) map defined by 
$\cT_{ABA'B'}(X)=\int_g {\rm d}g\: g^{\otimes 4}X(g^\dagger)^{\otimes 4}$, where $dg$ is the Haar measure on $U(d)$ normalised to $\int dg =1$. By Lemma~\ref{lem:rep} we have
\[
  \anti^{\otimes 2} \cong \sym^2(\anti) \oplus \wedge^2(\anti) 
                    \cong \left( \col \oplus \boxx\right) \oplus \hook,
\]
where $\col, \boxx$ and $\hook$ are irreducible representations of $U(d)$. 
It is furthermore remarkable that all irreducible representations 
have multiplicity at most one for general $d$. Such a case is called multiplicity-free and will 
be one of the main reasons why we can carry out our computation. 

By elementary representation theory we can pull this result to the $n$-fold systems and conclude that 
\[
  \Omega_{A^nB^nA'^nB'^n} = \sum_{y_1, \ldots, y_n} p_{y_1\ldots y_n} 
                                                      \rho_{y_1}\otimes  \cdots \otimes \rho_{y_n},
\]
where the constants $p_{y^n}$ are non-negative and sum to one, and 
$y_i \in \{\col, \boxx, \hook \}$ are indices keeping track in which irreducible representation we are
(denoted by their Young diagram). 
The states $\rho_y$ are proportional to the identity on the respective representation. 
The probability distribution can furthermore be taken to be invariant under 
permutation of the labels. 
Note also that the state $\ket{\psi_{A^nB^n}}\otimes \ket{\psi_{A'^nB'^n} }$ is invariant under $F_{A^nB^n:A'^nB'^n} = \bigotimes_{i=1}^n F_{A_i:A_i'} \otimes F_{B_i:B_i'}$, 
this implies $F_{A^nB^n:A'^nB'^n}\Omega = \Omega$. 
We now observe that $F_{A_i:A_i'} \otimes F_{B_i:B_i'}$ when restricted to the 
subspace corresponding to $\col$ and $\boxx$ acts as the identity, and when restricted 
to $\hook$ acts as minus the identity. In order to see this note that 
$F_{A:A'} \otimes F_{B:B'}$ acts trivially on 
$\text{Sym}^2(\anti)=\text{Sym}^2(\wedge^2(\complex^d))=\col \oplus \boxx$ 
and flips the sign on the orthogonal complement $\wedge^2(\wedge^2(\complex^d))$ 
which equals $\hook$. This shows that sequences $y^n$ with nonzero $p_{y^n}$ 
must have an even number of $\hook$'s. In summary, 
\[
  \Omega_{A^nB^nA'^nB'^n} = \sum_{y^n: \#  \hook \text{'s even}} p_{y_1\ldots y_n}
                                                    \rho_{y_1}\otimes  \cdots \otimes \rho_{y_n}.
\]
Note further that the state $\Omega_{A^nB^nA'^nB'^n}$ is of the form
\[
  \Omega_{A^nB^nA'^nB'^n}=\int \mu(\alpha) \proj{\alpha}_{A^nB^n}\otimes \proj{\alpha}_{A'^nB'^n} d\alpha
\]
for some probability density $\mu(\alpha)$ with respect to the Haar measure $d\alpha$. This state is therefore separable across the
$A^nB^n : A'^nB'^n$ cut. Note also that every separable state on 
$\sym^2(\anti^{\otimes n})$ takes this form. 
\end{proof}

\medskip
We have thus succeeded to transform the maximisation of the purity of the
reduced state over quantum states, which is a quadratic objective
function, in Eq.~\eqref{eq:Werner} to a linear optimisation problem over
finitely many non-negative real numbers, but with an
additional separability constraint, as given by Eq.~\eqref{eq:sep}.
Since this requirement of separability is difficult to handle we will now relax
the optimisation problem by only demanding
that the state should have a positive partial transpose (PPT). 

Since the PPT constraint, unlike separability, is a semidefinite
constraint, we are then dealing with a semidefinite programme,
and that duality theory should be able to give some information on the
maximum value -- see a similar line of argument
in~\cite{audenaert-2001-87}. As we will now show, the resulting problem
[obtained by relaxing Eq.~\eqref{eq:sep} to PPT states in the right hand
side optimization in Eq.~\eqref{eq:Werner}] is indeed a linear programme.
In order to do so we need to express the PPT condition as a linear
constraint in the variables $p_{y^n}$ and the target function as a linear
function in them. This is accomplished by the following lemma.

\begin{lemma} \label{lemma:LP}
$\displaystyle{\max_{\ket{\psi}_{A^nB^n} \in \anti^{\otimes n}}}  \tr \psi_{A^n}^2  \leq  \zeta_{n, d}$, where
\begin{equation}  
  \zeta_{n, d} := \max\, \vec{t}^{\,\ox n} \!\cdot \vec{p} \; \; \text{ s.t. } \vec{p} \geq 0,\  \vec{1}\cdot\vec{p} = 1,\ {T}_d^{\ox n}\vec{p} \geq 0. 
\label{eq:d-LP}
\end{equation}
Here, $\vec{t} = (-1, \half, 0)$, and the matrix $T_d$ is given by
\[
 {T}_d = 
 \left(\begin{array}{ccc}
                                  1 &      1 &             -1 \\
                  -2- \frac{6}{d-2} &      1 &  \frac{2}{d-2} \\
               1+\frac{2(d^2-d+1)}{d(d-1)(d-2)}
                                    & 1-\frac{d+1}{d(d-1)}
                                             & 1-\frac{2d-3}{d(d-1)(d-2)}
         \end{array}\right).
\]
\end{lemma}

\begin{proof} 
The objective function takes the form 
  \[\begin{split}
    \tr \Omega_{A^nB^nA'^nB'^n} (F_{A^n:A'^n} \otimes \1_{B^nB'^n})
       &= \tr \Omega_{A^nA'^n} F_{A^n:A'^n} \\
       &= \sum_{y^n \in \{ \col, \boxx, \hook\}^n} 
                p_{y_1\ldots y_n} \tr (\tilde\rho_{y_1}\otimes \cdots \otimes \tilde\rho_{y_n})F_{A^n:A'^n} \\
       &= \sum_{y^n \in \{ \col, \boxx, \hook\}^n} 
                p_{y_1\ldots y_n} \prod_{i=1}^n \tr \tilde\rho_{y_i} F_{A_i:A_i'}\\
       &= \sum_{y^n \in \{ \col, \boxx, \hook\}^n} 
                p_{y_1\ldots y_n} \prod_{i=1}^n t_{y_i}\\
       &= \vec{t}^{\ox n} \!\cdot \vec{p} 
  \end{split}\]
where we defined $\tilde\rho_{y}=\tr_{BB'}\rho_y$ and $t_y = \tr \tilde\rho_{y} F_{A:A'}$. The calculation of the coefficients $t_y$, which we arrange in the vector $\vec t:=(t_\col, t_\boxx, t_\hook)$ can be found in Lemma~\ref{lem:tildes} in Appendix~\ref{app:A}. 

We will now relax the constraints of the optimisation problem. As a first
step we remove the constraint that the number of $\hook$'s in a string
$y^n$ is even. As a second step we replace the separability constraint by
the weaker constraint that the state is PPT. The partial transposes of
$\rho_y$ with respect to the $AB:A'B'$ cut, denoted by $\rho_y^\Gamma$, are computed in
Appendix~\ref{app:A}. Since these $\rho_y^\Gamma$ commute with
all $g\ox g\ox\overline{g}\ox\overline{g}$, it is natural to first find
the decomposition of the space $\wedge^2(\complex^d) \otimes \wedge^2(\complex^d) \subset (\complex^d)^{\otimes 4}$ into the spaces of irreducible representations of $U(d)$ when $U(d)$ acts on $\wedge^2(\complex^d) \otimes \wedge^2(\complex^d) $ via its action $g\ox g\ox\overline{g}\ox\overline{g}$ on $ (\complex^d)^{\otimes 4}$. It turns out that the space has three components of multiplicity $1$ each, given
by projectors
\begin{align*}
  \Psi &= \proj{\Psi} \text{ for } \ket{\Psi}
        = \frac{1}{\sqrt{{d\choose 2}}}\sum_{i<j} \ket{\psi_{ij}}\ket{\psi_{ij}}, \\
  Q    &= \frac{2d}{d-2}(P_\anti \ox P_\anti)\bigl( (\1-\Phi)_{AA'} \ox \Phi_{BB'} \bigr) (P_\anti \ox P_\anti), \\
  \PP  &= P_\anti \ox P_\anti - Q - \Psi,
\end{align*}
having dimensions $1$, $d^2-1$ and $\left( \frac{d(d-1)}{2} \right)^2-d^2$,
respectively; see Lemma~\ref{lem:Psi-Q-PP} in Appendix~\ref{app:A}. Here, $\Phi$ denotes the maximally entangled state. Using the symmetries of the states and these projectors, it
is not hard to compute the overlap of all $\rho_y^\Gamma$
with each of the above (Lemma~\ref{lem:pt-overlaps} in Appendix~\ref{app:A}).
The result is
\begin{align*}
  \rho_\col^\Gamma  &= \frac{1}{{d\choose 2}} \Psi - \frac{2(d+1)}{d(d-2)} 
  \tilde{Q}
                                + \left( 1 + \frac{2(d+1)}{d(d-2)} - \frac{1}{{d\choose 2}} \right) \tilde\PP, \\
  \rho_\boxx^\Gamma &= \frac{1}{{d\choose 2}} \Psi + \frac{1}{d} \tilde{Q}
                                + \left( 1 - \frac{1}{d} - \frac{1}{{d\choose 2}} \right) \tilde\PP, \\
  \rho_\hook^\Gamma &=-\frac{1}{{d\choose 2}} \Psi + \frac{2}{d(d-2)} \tilde{Q}
                                + \left( 1 - \frac{2}{d(d-2)} + \frac{1}{{d\choose 2}} \right) \tilde\PP.
\end{align*}
Here, we defined $\tilde{Q}=Q/(d^2-1)$ and $\tilde\PP= \PP / ((d(d-1)/2)^2-d^2)$. We now introduce the matrix
\begin{equation}
  \label{eq:T}
  \hat{T}_d :=  \left(\begin{array}{ccc}
                  \frac{2}{d(d-1)} & \frac{2}{d(d-1)} & -\frac{2}{d(d-1)} \\
            -\frac{2(d+1)}{d(d-2)} &      \frac{1}{d} &  \frac{2}{d(d-2)} \\
            1+\frac{2(d+1)}{d(d-2)}-\frac{2}{d(d-1)}
                                   & 1-\frac{1}{d}-\frac{2}{d(d-1)}
                                                      & 1-\frac{2}{d(d-2)}+\frac{2}{d(d-1)}
          \end{array}\right),
\end{equation}
where the rows of the matrix are labelled by $\Psi$, $\tilde{Q}$ and $\tilde\PP$, 
and the columns of the matrix $\hat{T}_d$ are labelled by $\col$, $\boxx$
and $\hook$, in that order. The PPT constraint on the state $\Omega$ then
turns into the following linear constraints on the probability vector
$\vec p$\,:
\begin{equation} \hat{T}_d^{\ox n}\vec{p} \geq 0.
  \label{eq:hat-d-LP} 
\end{equation}
Without loss of generality, $p_{y^n}$ is permutation invariant.

A little later, we will take the limit $d\rightarrow \infty$. Observe therefore that some of the matrix entries of $T_d$ tend to zero as 
$d\rightarrow\infty$ and the linear programme would become trivial under this limit. For the linear programme, however, only the positivity
condition in eq.~(\ref{eq:hat-d-LP}) plays a role. This condition
remains unchanged if we choose a new operator basis
\[
  \frac{2}{d(d-1)}\Psi,\ \frac{1}{d}\tilde{Q},\ \tilde\PP,
\]
which transforms $\hat{T}_d$ into 
\[
  T_d = \left(\begin{array}{ccc}
                                     1 &      1 &             -1 \\
                   -\frac{2(d+1)}{d-2} &      1 &  \frac{2}{d-2} \\
               1+\frac{2(d+1)}{d(d-2)}-\frac{2}{d(d-1)}
                                       & 1-\frac{1}{d}-\frac{2}{d(d-1)}
                                                & 1-\frac{2}{d(d-2)}+\frac{2}{d(d-1)}
         \end{array}\right).
\]
This concludes the proof of the lemma.
\end{proof}

\medskip
As a corollary to Lemma~\ref{lemma:LP} we can already reproduce the result regarding $\alpha_3$:
\begin{corollary}[Yura~\cite{Yura:E_C}]
  \label{cor:yura}
  For all $n$, $E_F(\alpha_3^{\ox n}) = n$, hence $E_C(\alpha_3) = 1$.
\end{corollary}
\begin{proof}
  As mentioned earlier, the case $d=3$ is special because the irreducible representation $\col$
  is zero-dimensional, and hence doesn't appear in the above linear programme:
  $p_{y^n} = 0$ if any $y_i$ equals $\col$. But then the objective
  function of the linear programme (\ref{eq:d-LP}) is upper bounded by $2^{-n}$ 
  since that is the largest
  coefficient $t_{y^n}$, $y^n\in\{\boxx,\hook\}^n$ and $\sum_{y^n} p_{y^n}=1$.
  Thus, by Lemmas~\ref{lemma:renyi-2} and \ref{lemma:symmetrization},
  $E_F(\alpha_3^{\ox n}) \geq -\log_2 2^{-n} = n$, while the
  opposite inequality is trivial.
\end{proof}

\medskip
For $d\geq 4$ the irreducible representation $\col$ is present, and for all $y^n$
with an even number of it, the objective function of the linear programme
(\ref{eq:d-LP}) gets a contribution potentially larger than $2^{-n}$.
Motivated by the fact that (thanks to the LOCC monotonicity of $E_F$ under
twirling) $E_F(\alpha_d^{\ox n})$ monotonically decreases with $d$,
we aim to understand this linear programme for fixed $n$ but asymptotically
large $d$. Note that in the limit $d \rightarrow \infty$, the matrix $T_d$ converges to
\[
  {T}_\infty = \left(\begin{array}{rrr}
                 1 & 1 & -1 \\
                -2 & 1 &  0 \\
                 1 & 1 &  1
             \end{array}\right).
\]
Thus we find that $E_F(\alpha_d^{\otimes n})$
for fixed $n$ and arbitrary $d$ is lower bounded by $-\log_2 \zeta_n$, where 
\begin{equation}\begin{split}
  \zeta_n:= \max\, \vec{t}^{\ox n} \cdot \vec{p} 
  \quad \text{s.t. }       \vec{p}                             &\geq 0, \\
                                         \vec{1} \cdot \vec{p} &=    1, \\
                       {T}_\infty^{\ox n}\vec{p}               &\geq 0.
  \label{eq:infty-LP}
\end{split}\end{equation}
with the additional constraint that $p_{y^n}$ is permutation invariant.

From the linear programme we now eliminate all constraints that involve the first row of $T_\infty$, thereby only increasing the value of the linear programme. Mathematically, we delete the first row of $T_\infty$ and now have
\[
  \left(\begin{array}{rrr}
     -2 & 1 & 0 \\  1 & 1 & 1 
      \end{array}\right).
\]
We then see that we do not need to consider vectors $y_n$ which contain one $\hook$
or more. Namely, in the expansion of the state $\Omega$ every single
occurrence of $\rho_\hook$ may be replaced with 
$\frac{1}{3}\rho_\col + \frac{2}{3}\rho_\boxx$,
turning a feasible point into a new feasible point, and not
changing the value of the objective function.
But then, since the entries of the last column are never used again in the constraints,
we may delete it leaving a truncated matrix and a truncated vector
\[
   T := \left(\begin{array}{rr}
          -2 & 1 \\  1 & 1 
       \end{array}\right), \quad
   \vec{t} = \left( -1, \ \tfrac12 \right)\ .
\]
(Note that we may relax the normalization condition $\vec{1}\cdot\vec{p} = 1$
w.l.o.g.~to $\leq 1$.)
\begin{corollary}
  \label{prop:simpler-LP}
  For any $d$ and $n$, $E_F(\alpha_d^{\ox n}) \geq -\log_2 \zeta_n$, where
  \begin{equation}\begin{split}
    \zeta_n = 
    \max\, \vec{t}^{\,\ox n} \cdot \vec{p} = 2^{-n}\sum_{y^n\in\{\col, \boxx \}^n} p_{y^n}(-2)^{|y^n|}
    \quad \text{s.t. }          \vec{p} &\geq 0, \\
                  \vec{1} \cdot \vec{p} &\leq 1, \\
                      -T^{\ox n}\vec{p} &\leq 0,
    \label{eq:simpler-LP}
  \end{split}\end{equation}
  where $p_{y^n}$ only depends on the number $|y^n|$ of occurrences of $\col$.
  
  Note that in this form the linear programme does not refer to $d$ any more; 
  it reflects the limit $d \rightarrow \infty$ completely.
  \qed
\end{corollary}

Now, all that is left to do is to find an upper bound on $\zeta_n$, which
we obtain by writing down the dual linear programme~\cite{LP-book} and
guessing a dual feasible point.

\begin{lemma} \label{lemma:threequarters}
$\zeta_n\leq (\frac{3}{4})^n$, hence $E_F(\alpha_d^{\ox n}) \geq n \log_2\frac{4}{3}$.
\end{lemma}
\begin{proof}
The dual linear programme to~\eqref{eq:simpler-LP} is given by 
\begin{equation}
  \min\, z  \quad \text{ s.t. }  \vec{q} \geq 0,\   z\vec{1} - S^{\ox n}\vec{q} \geq \vec{t}^{\,\ox n},
  \label{eq:simpler-dual-LP}
\end{equation}
where $ S = T^\top $ and $\top$ denotes matrix transposition. Its value equals $\zeta_n$ by linear programming duality. 

In words, a feasible $z$ in the dual linear programme is an upper bound on all
the vector entries of $\vec{t}^{\,\ox n}+S^{\ox n}\vec{q}$. (Caution: 
some of these may be negative, and so we are not talking about the 
sup-norm of this vector.) By duality, any such $z$ is going to be an
upper bound on $\zeta_n$~\cite{LP-book}.

The entries of $\vec{q}$ are labelled by strings $w^n \in \{\Psi,Q\}^n$,
and it is clear from the permutation symmetry of the matrix $S^{\ox n}$ and 
the vector $\vec{t}^{\,\ox n}$ that we may assume that $q_{w^n}$ only
depends on the number $k$ of $Q$'s in $w^n$:
\[
  \delta_k := q_{\Psi^{n-k}Q^k} \text{ and all permutations, for } k=0,\ldots,n.
\]
Then, also the constraints in the dual linear programme (\ref{eq:simpler-dual-LP}), which
are labelled by strings $v^n \in\{0,1\}^n$, depend only on the number $m$ of
$0$'s: for each string $v^n=0^{m}1^{n-m}$, $m=0,\ldots,n$, we get an
inequality
\begin{equation}
  z \geq (-1)^m\, 2^{m-n} + \sum_{k=0}^n \delta_k 
                                          \sum_{\ell=\max(0,k+m-n)}^{\min(k,m)}
                                           (-2)^\ell {m \choose \ell} {n-m \choose k-\ell}.
  \label{eq:dual-symmetrised}
\end{equation}

Numerical solutions of the linear programme (\ref{eq:dual-symmetrised}) suggest 
that in the dual only $\delta_1$ is populated and
the $\delta_j$ with $j \approx n$. 
Here we guess a dual feasible solution motivated by this. The ansatz
is only an approximation to the numerical findings; for some non-negative
$\beta<1$ and $\gamma$,
\begin{align*}
  \delta_k &= \gamma \beta^{n-k},\ \text{ for }k < n \\
  \delta_n &= 0.
\end{align*}
Clearly, all $\delta_j$ are now nonnegative;
inserting the above into the dual constraint~(\ref{eq:dual-symmetrised})
yields, for all $m$, that
\[
  z \geq (-2)^m 2^{-n} + \sum_{k=0}^n \gamma \beta^{n-k}
                                        \sum_{\ell=\max(0,k+m-n)}^{\min(k,m)}
                                        (-2)^\ell {m \choose \ell} {n-m \choose k-\ell}
                       - \gamma (-2)^m,
\]
noticing that the coefficient of the variable $\delta_n$ in
eq.~(\ref{eq:dual-symmetrised}) is $(-2)^m$.
First we evaluate the double sum; observe that it involves all pairs
of $k$ and $\ell$ for which the binomial coefficients are nonzero.
Hence, it is
\[\begin{split}
  \sum_{k,\ell} \gamma \beta^{n-k} (-2)^\ell {m \choose \ell} {n-m \choose k-\ell}
     &= \sum_{k,\ell} \gamma \beta^{n-(k-\ell)-\ell} (-2)^\ell {m \choose \ell} {n-m \choose k-\ell} \\
     &= \gamma \beta^n \sum_{k,\ell} \left(\frac{1}{\beta}\right)^{k-\ell} \left(\frac{-2}{\beta}\right)^\ell {m \choose \ell} {n-m \choose k-\ell} \\
     &= \gamma \beta^n \left( 1+\frac{1}{\beta}\right)^{n-m} \left( 1-\frac{2}{\beta}\right)^m \\
     &= \gamma (\beta+1)^{n-m} (\beta-2)^m.
\end{split}\]
This simplifies the constraints to
\[
  \forall m \quad z \geq (-2)^m\left( 2^{-n} - \gamma \right) + \gamma (\beta+1)^{n-m} (\beta-2)^m,
\]
so $z$ is the maximum of the right hand side over all $m=0,\ldots,n$, and we
want to choose $\beta$ and $\gamma$ in an optimal way to minimize this
maximum. First of all, the first term can grow very large due to the
occurrence of $2^m$ -- so the only reasonable choice is $\gamma = 2^{-n}$.
This reduces the constraints to
\[
  \forall m \quad z \geq 2^{-n}(1+\beta)^n (-1)^m \left( \frac{2-\beta}{1+\beta} \right)^m,
\]
so choosing $\beta = 1/2$, and neglecting the signs, makes the right
hand side $(3/4)^n$.

In conclusion, we obtain a dual feasible solution with this value,
yielding an upper bound $\zeta_n \leq (3/4)^n$,
which gives this as an upper bound on the maximum purity of a 
reduced state in $n$ copies of the antisymmetric subspace.
\end{proof}

\medskip
Theorem~\ref{th:EC-lower} is now a direct consequence of Lemma~\ref{lemma:threequarters}.

\section{Regularised Relative Entropy of Entanglement}
\label{sec:relent}
Here we show that the constant lower bound on the
entanglement cost of the antisymmetric state that we have calculated above
implies a constant lower bound on the regularised relative entropy of
entanglement with respect to separable states
[Eq.~\eqref{def:relentasympt}],
\begin{equation}
        \label{eq:cor3-restated}
E_{R,\mathrm{sep}}^\infty(\alpha_d) \geq \log_2\sqrt{\frac{4}{3}}\gtrsim
0.2075\ ,
\end{equation}
as stated in Corollary~\ref{cor:ERinf_lowerbnd}. 


\begin{proof}\textbf{of Corollary \ref{cor:ERinf_lowerbnd}.}
We want to prove the lower bound \eqref{eq:cor3-restated} of
Corollary~\ref{cor:ERinf_lowerbnd}.  Since $\alpha_d$ is invariant under
$g\otimes g$ (for unitary $g$), the minimisation in the relative entropy
can be taken over states obeying the same symmetry condition, i.e.
$$E_{R,\mathrm{sep}}(\alpha_d^{\ox n})=\min D(\alpha_d^{\ox n}||\sigma),$$
where $\sigma$ is separable and $\sigma=\sum_{y^n \in \{0, 1\}^n} p_{y^n} \rho_{y_1} \otimes \cdots \otimes \rho_{y_n}$ for $\rho_{0}=\alpha_d$ and $\rho_{1}=\sigma_d$. The relative entropy evaluates in this case to 
$$\tr \alpha_d^{\ox n} \log_2 \alpha_d^{\ox n} - \tr \alpha_d^{\ox n}
\log_2 p_{00\cdots 0} \alpha_d^{\ox n}= -\log_2 p_{00\cdots 0}.$$ In
summary, $E_{R,\mathrm{sep}}(\alpha_d^{\ox n}) = -\log_2 \max\limits_\sigma
\tr\sigma
P_{\anti}^{\ox n}$, where the maximum is over states $\sigma$ separable
across $A^n:B^n$. Furthermore,
\[\begin{split}
  \max_{{\sigma \text{ separable} \atop \text{across }A^n:B^n}} \tr\sigma P_{\anti}^{\ox n}
       &= \max_{\ket{\alpha}\in A^n,\, \ket{\beta}\in B^n} 
                    \bra{\alpha}\bra{\beta} P_{\anti}^{\ox n} \ket{\alpha}\ket{\beta}    \\
       &= \max_{\ket{\alpha}\in A^n,\, \ket{\beta}\in B^n,\, \ket{\psi}\in\anti^{\ox n}}
                    \bigl| \bra{\alpha}\bra{\beta} \psi \rangle \bigr|^2                 \\
       &= \max_{\ket{\psi}\in\anti^{\ox n}}
                    \bigl\| \tr_{B^n} \proj{\psi} \bigr\|_\infty,
\end{split}\]
where the first equality is by convexity, the second by choosing $\ket{\psi}$ 
as the projection of $\ket{\alpha}\ket{\beta}$ into $\anti^{\ox n}$,
and the third by the Schmidt decomposition.
The expression in the last line is upper bounded by the 
square root of the maximum purity, which we showed above to be smaller or equal to $(3/4)^n$.
Hence, $E_{R,\mathrm{sep}}(\alpha_d^{\ox n}) \geq n \log_2\sqrt{\frac{4}{3}}$,
and we get the constant lower bound of $\log_2\sqrt{\frac{4}{3}} \approx 0.2075$
for $E_{R, \mathrm{sep}}^\infty(\alpha_d)$.
\end{proof}
In contrast, the calculation of~\cite{audenaert-2001-87} gave
$E_{R,\mathrm{PPT}}^\infty(\alpha_d) = \log_2\frac{d+2}{d}$. This shows, in particular, that $E_{R,\mathrm{PPT}}^\infty$ differs from $E_{R,\mathrm{sep}}^\infty$ on Werner states. We conclude that squashed entanglement can be much smaller than the regularised relative entropy of entanglement with respect to separable states; the opposite separation was known thanks to the ``flower states'' of~\cite{Horodecki2005}.

\medskip
We close this section by showing the asymptotic continuity of the regularised 
relative entropy of entanglement .
\begin{proposition}[\cite{christandlPhD}]
\label{prop-rel-ent-cont}
The regularised relative entropy of entanglement $E^{\infty}_{R,\mathrm{sep}}$ 
is asymptotically continuous, i.e.~there is a function
$\delta(\epsilon)$ with $\delta(\epsilon) \rightarrow 0$ for
$\epsilon \rightarrow 0$ such that for all $||\rho-\sigma||_1 \leq
\epsilon$
$$ |E^\infty_{R,\mathrm{sep}}(\rho)-E^\infty_{R,\mathrm{sep}}(\sigma)|\leq \delta(\epsilon) \log
d,$$ where $d$ is the dimension of the system supporting $\rho$
and $\sigma$. In fact the proof shows that $\delta(\epsilon)$ can be taken
as $2(\epsilon+h(\epsilon))$, where $h$ denotes the binary entropy function.
The same statement is true for the regularised relative entropy of entanglement with respect to PPT states, $E^\infty_{R,\mathrm{PPT}}$.
\end{proposition}

\begin{proof}
Let $\|\rho- \sigma\|_1 = \epsilon>0$, where $\rho$ and $\sigma$
are $d$-dimensional states. According to Alicki and Fannes
\cite{AliFan04}, there are states $\gamma$, $\tilde{\rho}$ and
$\tilde{\sigma}$ with $\gamma=(1-\epsilon)\rho+\epsilon
\tilde{\rho}=(1-\epsilon) \sigma+\epsilon \tilde{\sigma}$. If we
succeed to prove asymptotic continuity on mixtures, i.e.~
\begin{equation}
        \label{eq-rel-cont-mixt} |E^\infty_{R,\mathrm{sep}}(\rho)-E^\infty_{R,\mathrm{sep}}(\gamma)|
\leq  \frac{\delta(\epsilon)}{2} \log d,
\end{equation} 
then continuity for $\rho$ and
$\sigma$ follows by use of the triangle inequality:
$$|E^\infty_{R,\mathrm{sep}}(\rho)-E^\infty_{R,\mathrm{sep}}(\sigma)| \leq |E^\infty_{R,\mathrm{sep}}(\rho)-E^\infty_{R,\mathrm{sep}}(\gamma)|+|E^\infty_{R,\mathrm{sep}}(\gamma)-E^\infty_{R,\mathrm{sep}}(\sigma)| \leq \delta(\epsilon) \log d.$$
The main step in the proof of the
estimate~(\ref{eq-rel-cont-mixt}) is the following inequality for
an ensemble $\{ p_i, \tau_i\}$,
\begin{equation} \label{ineq-relent-ineq}
        \sum_i p_i E_{R,\mathrm{sep}}(\tau_i) - E_{R,\mathrm{sep}}\left(\sum_i p_i \tau_i\right) 
          \leq H\left(\sum_i p_i \tau_i\right) -\sum_i p_i H(\tau_i) \leq H\left(\sum_i p_i \proj{i}\right),
\end{equation}
where $\ket{i}$ denotes an orthonormal basis.
Inequality~(\ref{ineq-relent-ineq}) has first been proven for the
relative entropy with respect to the set of separable
states~\cite{linden05} (see also \cite{EFPPW00}) and then been extended to hold for any convex set that includes
the maximally mixed state~\cite{SynHor05}. Here, it implies the
following estimate
$$E_{R,\mathrm{sep}}(\gamma^{\otimes N}) \geq \sum_k \epsilon^k
(1-\epsilon)^{N-k} \binom{N}{k} E_{R,\mathrm{sep}}(\rho^{\otimes (N-k)} \otimes
\tilde{\rho}^{\otimes k})- N h(\epsilon), $$ where $h(\epsilon)$
is the Shannon entropy of the distribution $(\epsilon,
1-\epsilon)$. We will now replace all $\tilde{\rho}$'s on the RHS
by $\rho$'s. This is done in two steps: i) remove the states of
the form $\tilde{\rho}$ on the RHS, since the partial trace
operations is an LOCC operation the RHS can only decrease, ii)
append the states $\rho$ and apply the inequality
$$ E_{R,\mathrm{sep}}(\rho^{\otimes N}) \leq E_{R,\mathrm{sep}}(\rho^{\otimes (N-k)})+kE_{R,\mathrm{sep}}(\rho),$$
which holds by subadditivity of $E_{R,\mathrm{sep}}$. This gives
\begin{eqnarray*} E_{R,\mathrm{sep}}(\gamma^{\otimes N})
        &\geq& \sum_k \epsilon^k (1-\epsilon)^{N-k} \binom{N}{k}
        E_{R,\mathrm{sep}}(\rho^{\otimes (N-k)} \otimes \tilde{\rho}^{\otimes k})- N h(\epsilon) \\
        &\stackrel{i)}{\geq}& \sum_k \epsilon^k (1-\epsilon)^{N-k} \binom{N}{k}
        E_{R,\mathrm{sep}}(\rho^{\otimes (N-k)} )- N h(\epsilon) \\
        &\stackrel{ii)}{\geq}& \sum_k \epsilon^k (1-\epsilon)^{N-k} \binom{N}{k}
        (E_{R,\mathrm{sep}}(\rho^{\otimes N})-k E_{R,\mathrm{sep}}(\rho) )- N h(\epsilon) \\
        &=& E_{R,\mathrm{sep}}(\rho^{\otimes N}) - \sum_k k \epsilon^k
        (1-\epsilon)^{N-k} \binom{N}{k} E_{R,\mathrm{sep}}(\rho)- N h(\epsilon) \\
        &=& E_{R,\mathrm{sep}}(\rho^{\otimes N}) - N\epsilon E_{R,\mathrm{sep}}(\rho)- N h(\epsilon) \\
        &\geq& E_{R,\mathrm{sep}}(\rho^{\otimes N}) - N (\epsilon \log d+ h(\epsilon) ) \\
        &\geq& E_{R,\mathrm{sep}}(\rho^{\otimes N}) - N (\epsilon + h(\epsilon))\log d.
\end{eqnarray*}
The last equality sign is the evaluation of the mean value of the
binomial distribution. Since the above calculation holds for all
$N$, this shows
$$E_{R,\mathrm{sep}}^\infty(\gamma) \geq E_{R,\mathrm{sep}}^\infty(\rho)-\frac{\delta(\epsilon) }{2}\log d$$
for $\delta(\epsilon) :=2(\epsilon + h(\epsilon))$. Conversely, the
convexity of $E_{R,\mathrm{sep}}^\infty$ \cite{DoHoRu02} implies
$$E_{R,\mathrm{sep}}^\infty(\gamma) \leq (1-\epsilon) E_{R,\mathrm{sep}}^\infty(\rho)+\epsilon
E_{R,\mathrm{sep}}^\infty (\tilde{\rho}) \leq E_{R,\mathrm{sep}}^\infty(\rho) +\epsilon \log
d.$$ This concludes the proof of the
estimate~(\ref{eq-rel-cont-mixt}) and the proposition. The exact same reasoning applies to $E^\infty_{R,\mathrm{PPT}}$.
\end{proof}

\medskip
A vital ingredient in the proof was
inequality~(\ref{ineq-relent-ineq}), which bounds the strength of
the convexity of the relative entropy. Prior to this work, the
same inequality has been used in~\cite{HHHO05b} to prove 
that the relative entropy of entanglement cannot be locked. As both entanglement of
purification and formation are lockable, a simple translation of
inequality~(\ref{ineq-relent-ineq}) to these measures is not
possible. Other ways to verify that entanglement
cost under LOCC and LOq (local operations with a sublinear amount of quantum communication) are asymptotically continuous will have to be found.

\section{Conclusion}
\label{sec:conclusion}
We have shown a way of -- in principle -- calculating the R\'{e}nyi-2 entropic
version of the entanglement of cost of the $d\times d$-antisymmetric state
via convex optimisation and more specifically, semidefinite programming.
Using a linear programming relaxation we showed a constant lower bound,
independent of $d$. Tighter relaxations are possible, in principle 
capable of obtaining the exact value of the maximum purity of the
reduced state over all $\ket{\psi} \in \anti^{\ox n}$: in addition 
to the PPT condition of the state between $AB$ and $A'B'$, we should
impose that the state is shareable (or extendible) to more 
parties~\cite{dohertycomplete, Ioannou2007, oneandahalf, navascues1, navascues2}.
At the same time, we could show that the squashed 
entanglement of these states is asymptotically small, implying that also
their distillable key is asymptotically small.

We believe that our result is the strongest indication so far that
``quantum bound key'' exists: states with positive key cost to create them
(a notion not yet defined in the literature, and a little tricky
to formalize cleanly), while their distillable key is zero. At least
we show that the states have asymptotically vanishing distillable key
(it cannot be zero, as a lower bound of $\frac{1}{d}$ on $E_D$
is known); on the other hand, their entanglement cost does not vanish.

The technique to obtain the lower bound on $E_C(\alpha_d)$ is yet another
demonstration of the power of symmetry in entanglement theory; but to our
knowledge, with this work we provide first application of plethysms in this field. Unfortunately,
we could not prove the conjectured $E_C(\alpha_d) = 1$ as our PPT relaxation
cannot give anything better than $\approx 0.45$ as computer solutions of the
linear programme up to $n=12$ show (see Appendix~\ref{app:B}).
It remains to be investigated whether further constraints, for instance of shareability, can improve
the lower bound to $1$, or whether $E_C(\alpha_d) < 1$ holds. The latter would provide the first explicit counterexample to additivity.

In comparison to the large gap observed between the entanglement of formation and distillable key~\cite{generic-entanglement}, our work exhibits three advantages. Firstly, our example is constructive, secondly, we show that the distillable key can be made arbitrarily small and thirdly, we 
consider the entanglement cost, which is the right measure to compare with the distillable key, and which can be strictly smaller than the entanglement of formation~\cite{Hastings}. The distinction between entanglement cost and entanglement of formation is crucial here, as it was for the discovery of bound entanglement~\cite{boundentanglement}, since the asymptotic measure of distillable key has to be compared to an asymptotic measure of preparing the state. A further result in~\cite{generic-entanglement} shows that the one-way distillable key is generically small, even if entanglement of formation is large. In our work, in contrast, the one-way distillable key of the antisymmetric state $\alpha_d$ vanishes for all $d\geq 3$. 

Our results can readily be generalised to the multiparty entanglement of the state proportional to the antisymmetric projector onto several parties. The multiparty squashed entanglement and distillable key~\cite{yang} exhibit a behaviour similar to the two-party case. Due to the difficulty of classifying multiparty entanglement, it is not clear which multiparty generalisation of entanglement cost to use. Any such generalisation, however, should be larger than entanglement cost of the two-party state, to which our lower bound applies.

\section*{Acknowledgments}
After completion of this work, F.~Brand\~{a}o kindly pointed out to us that the states from~\cite{Horodecki2005} can be used to construct states with $E_C(\rho)\geq \half$ and $K_D(\rho)\leq \frac{2}{\log_2 d}$. 
MC was supported by the Swiss National Science Foundation
(grant PP00P2-128455), the National Centre of Competence in Research 'Quantum Science and Technology' and the German Science Foundation (grants \mbox{CH~843/1-1} and \mbox{CH~843/2-1}).
NS acknowledges support by the EU (QUEVADIS,
SCALA), the German cluster of excellence project MAP,
the Gordon and Betty Moore Foundation through Caltech’s Center for the
Physics of Information, and the NSF Grant No. PHY-0803371.
AW is supported by the European Commission, the
U.K.~EPSRC, the Royal Society, and a Philip Leverhulme Prize.
The Centre for Quantum Technologies is funded by the
Singapore Ministry of Education and the National Research Foundation
as part of the Research Centres of Excellence programme.

\begin{appendix}

\section{Representation Theory}
\label{app:A}
Here we review certain facts of the representation theory of $U(d)$,
the unitary group in dimension $d$, particularly related to plethysms.
For the basic concepts we refer the reader to textbooks such as~\cite{FultonHarris91}.
The concatenation of two representations is a called a plethysm. In our case, we consider a representation $V_\mu$ of $U(d)$ and concatenate it with a representation $V_\lambda$ of $U(\dim V_\mu)$ to yield the $U(d)$-representation
$$V_\lambda (V_\mu): g \mapsto V_\lambda (V_\mu(g)).$$
\begin{lemma}
\label{lem:rep} 
Let $d\geq 3$. The following two plethysms of $U(d)$ decompose into 
irreducible representations of $U(d)$ as follows:
\begin{align*}
  \sym^2(\wedge^2)   &\cong\col \oplus \boxx, \\
  \wedge^2(\wedge^2) &\cong\hook.
\end{align*}
The dimensions are given by
\begin{align*}
\dim \sym^2(\wedge^2)&=\frac{d(d-1)(d^2-d+2)}{8}, \\
\dim \col &=\frac{d(d-1)(d-2)(d-3)}{24},          \\
\dim \boxx &=\frac{(d+1)d^2(d-1)}{12},             \\
\dim \wedge^2(\wedge^2)=\dim \hook&=\frac{(d+1)d(d-1)(d-2)}{8}.
\end{align*}
Note that $\dim \col =0$ for $d=3$.
\end{lemma}
\begin{proof}
We will compute the decomposition of the representations by a decomposition of the corresponding characters.
The character of an irreducible representation of $U(\ell)$ with highest weight $\lambda$ is given by
\begin{equation} 
  \label{eq:Schur} 
  s_\lambda(z_1, \ldots, z_\ell)=\sum_{T} z_{T(1)}\cdots z_{T(\ell)},
\end{equation}
where the sum extends over all semi-standard Young tableaux of shape 
$\lambda$ with numbers $1, \ldots, \ell$, that is, over all fillings of the 
boxes of the Young diagram $\lambda$ with the numbers $1, \ldots, \ell$ such 
that they strictly decrease downwards and decrease weakly to the right. 

The characters of $\sym^2$ and $\wedge^2$ as representations of $U(\ell)$ are
$$s_{\sym^2}(z_1, \ldots, z_\ell)=\sum_{i\leq j} z_iz_j$$
$$s_{\wedge^2}(z_1, \ldots, z_\ell)=\sum_{i < j} z_iz_j.$$
Reducing the $U(\ell)$ representation, where $\ell=\frac{d(d-1)}{2}$ to a representation of $U(d)$ via its action on $\wedge^2$ corresponds to making the replacement $z_i \mapsto x_{k}x_l$, where $1\leq k < l\leq d$. Hence
$$s_{\sym^2(\wedge^2)}(x_1, \ldots, x_d)=s_{\sym^2}(x_1x_2, \ldots, x_{d-1}x_d)=\sum_{k<l, m<n, (kl)\leq (mn) } x_k x_l x_m x_n $$
The summation can be rewritten as
$k<l, m<n, k<m, l\leq n \text{ or } k<l, m<n, k<m, l> n \text{ or } k<l, m<n, k=m, l\leq n $
which can be condensed to 
$k<l, m<n, k\leq m,  l\leq n \text{ or } k<m<n< l $, and which results in the decomposition  
$$s_{\sym^2(\wedge^2)}(x_1, \ldots, x_d)=s_{\boxx}(x_1, \ldots, x_d)+s_{\col}(x_1, \ldots, x_d)$$
by use of Eq.~\eqref{eq:Schur}.
The second character takes the form
$$s_{\wedge^2(\wedge^2)}(x_1, \ldots, x_d)=s_{\wedge^2}(x_1x_2, \ldots, x_{d-1}x_d)=\sum_{k<l, m<n, (kl)<(mn)} x_k x_l x_m x_n.$$
The summation can be rewritten as
$k<l, m<n, k<m \text{ or } k<l, m<n, k=m, l<n$
which is equivalent to 
$k<l, k<m<n \text{ or } k=m, k<l<n .$
Relabeling in the second clause $m \leftrightarrow l$, we can combine both clauses to 
$k\leq l, k<m<n. $ Hence, we obtain
$s_{\wedge^2(\wedge^2)}(x_1, \ldots, x_d)=\sum_{k\leq l, k<m<n} x_k x_l x_m x_n=s_{\hook}(x_1, \ldots, x_d)$
where the latter equation follows from Eq.~\eqref{eq:Schur}.
The lemma follows since the decomposition of the characters is unique and in one-to-one relation with the decomposition of the representations themselves. The dimensions are computed with help of Weyl's dimension formula, equation~\eqref{eq:Weyldim}.
\end{proof}

\begin{lemma} 
Let $d\geq 3$. The projectors onto the subspaces $\col, \boxx$ and $\hook$ embedded into $\sym^2(\wedge^2)$ and $\wedge^2(\wedge^2)$, both embedded into $ABA'B'$ as in Lemma~\ref{lem:rep} are given by
\begin{align}
P_{\col}&=\frac{1}{24} \sum_{\pi \in S_4}\sign(\pi) \pi\\
P_{\boxx}&=\frac{1}{48} \left(e-(12)\right) \left(e-(34)\right) \left( e+(13)\right) \left( e+(24) \right) \left(  e-(12) \right)\left(  e-(34) \right) \\
P_{\hook}&=\frac{1}{4} \left( e-(12) \right) \left( e-(34) \right) - P_{\col}-P_{\boxx}.
\end{align}
where the order of the systems is $ABA'B'$.
\end{lemma}
\begin{proof}
All three representations are subrepresentations of $g \mapsto g^{\otimes 4}$ which decomposes, according to Schur-Weyl duality, into irreducible representations in the following way (for $d=3$, $\col$ does not appear):
$$ \col \oplus 3 \hook \oplus 2 \boxx \oplus 3\,\yng(3,1) \oplus \,\yng(4).$$
The isotypical subspaces can be constructed with help of Young projectors which are proportional to the formula (for $\lambda$ being one of the five irreducible representations)
$$Q_\lambda =\sum_T Q_T$$
where the sum goes overall all standard tableaux of shape $\lambda$ with numbers $1, \ldots, 4$ and where 
$$Q_T = \left( \sum_{\pi \in \cC(T)} \sign(\pi)\pi \right) \left( \sum_{\pi \in \cR(T)} \pi  \right)  $$
is proportional to the projector onto one copy of an irreducible representation with highest weight $\lambda$. 
From this we can readily verify the above formula for $\col$. For $\boxx$ we make the guess $T=
\begin{Young}
1&3\cr
2&4\cr
\end{Young}
$
and are lucky: since the corresponding space is antisymmetric when we exchange $1$ and $2$ and also when we exchange $3$ and $4$ it is contained in $(\wedge^2)^{\otimes 2}$. The projector onto $\hook$ follows from observing that the projector onto $(\wedge^2)^{\otimes 2}$ is given by $\frac{1}{4} \left( e-(12) \right) \left( e-(34) \right)$ and that all three, $\col, \boxx$ and $\hook$, have to add to this space.
\end{proof}

We define the corresponding quantum states by 
\begin{align}
\rho_{\col}&=\frac{24}{d (d-1) (d-2) (d-3)}P_{\col}, \\
\rho_{\boxx}&=\frac{12}{(d+1)d^2(d-1)}P_{\boxx},     \\
\rho_{\hook}&=\frac{8}{(d+1)d(d-1)(d-2)}P_{\hook}.
\end{align}

\begin{lemma} 
\label{lem:tildes}
$$\vec t:=(t_\col, t_\boxx, t_\hook)=(-1, \half, 0),$$
where $t_y = \tr \tilde\rho_{y} F_{A:A'}$ and $\tilde\rho_{y}=\tr_{BB'}\rho_y$. Equivalently, we can write
\begin{align}
  \tilde\rho_{\col}  &= \alpha, \\
  \tilde\rho_{\boxx} &= \frac{1}{4} \alpha + \frac{3}{4}\sigma, \\
  \tilde\rho_{\hook} &= \half \alpha + \half \sigma.
\end{align}
where $\sigma$ and $\alpha$ are proportional to the projectors onto the
symmetric and antisymmetric subspace, respectively.
\end{lemma}
\begin{proof}
Since all three states commute with the action of $g \otimes g$ ($g \in U(d)$) they are Werner states and thus of the form $p \alpha +(1-p) \sigma$ for $0\leq p \leq 1$. Note that the $p_i$ satisfy the equation $1-2p_i=\tr \tilde\rho_{i} F_{AA'}=\tr \rho_{i} (F_{AA'}\otimes \1_{BB'})$. 

We will now verify the claim state by state: The state $\tilde\rho_{\col}$ is the partial trace over a totally antisymmetric state and thus totally antisymmetric itself, hence $p_{\col}=1$ and thus $t_\col=-1$.

The state $\rho_\boxx$ is the normalisation of the projector
\begin{align*}
 P_{\boxx} &=\frac{1}{24} \big(  2 e-2 (12) - 2(34) + (13) +(14) + (23) + (24) + 2 (12)(34) + 2 (13)(24) + 2 (14)(23)\\
 &\quad  - (123) - (132) - (124) - (142) - (134) - (143) - (234) - (243) \\
 &\quad + (1234) + (1243) + (1342) + (1432) -2 (1324) - 2(1423) \big).
\end{align*}
Multiplying it from the right with the flip operator results in 
\begin{align*}
 P_{\boxx} (F_{AA'}\otimes \1_{BB'})
     &= P_{\boxx} (13) \\
     &\!\!\!\!\!\!\!\!\!\!\!\!\!\!\!\!\!\!\!\!\!\!\!\!\!\!\!\!
      = \frac{1}{24} \big(  2 (13) -2 (132) - 2(142) + e + (134) + (123) + (13)(24) + 2 (1432) + 2 (24) + 2 (1234)\\
     & - (23) - (12) - (1324) - (1342) - (14) - (34) - (1423) - (1243) \\
     & + (14)(23) + (243) + (142) + (12)(34) -2 (124) - 2 (234) \big).
\end{align*}
We now take the trace of this equation and find, since the trace of a cycle equals $d$,
$t_\boxx=\tr \rho_{\boxx} F_{AA'}\otimes \1_{BB'}=\half$ or $p_{\boxx}=\frac{1}{4}$.

Finally, $t_\hook$ is proportional to 
\begin{align*}
  \tr P_{\hook} (F_{AA'}\otimes \1_{BB'}) 
          &=\left(\frac{d(d-1)}{2}\right)^2 \tr (\alpha_{AB} \otimes \alpha_{A'B'})(F_{AA'}\otimes \1_{BB'}) \\
          &\qquad 
           -\tr P_{\col} F_{AA'}\otimes \1_{BB'}-\tr P_{\boxx} F_{AA'}\otimes \1_{BB'} \\
          &=\left(\frac{d(d-1)}{2}\right)^2  \frac{d}{d^2} 
           - (-1) \frac{d(d-1)(d-2)(d-3)}{24} -\half \frac{(d+1)d^2(d-1)}{12} = 0.
\end{align*}
This implies $p_{\hook}=\half$ and concludes the proof.
\end{proof}

Next we derive some formulas regarding the partial transposes of the
states $\rho_y$, $y \in \{ \col,\boxx,\hook\}$ with respect to the
$AB:A'B'$ cut. Due to the partial transpose we have to deal with
decomposing tensor products that involve dual representations. In order to
be able to continue to use the Young frame notation (rather than the
highest weight notation) in this situation, we use $SU(d)$ rather than
$U(d)$. The action of $SU(d)$ on $\wedge^d(\complex^d)$ is namely trivial
and allows us therefore to add full columns and convert negative weights
into positive ones. For the spaces, this difference is immaterial and
therefore of no concern to us.

\begin{lemma}
\label{lem:Psi-Q-PP}
The decomposition of the representation
$\anti \ox \overline{\anti}$ of $SU(d)$ is given by
\[
  \anti \otimes \overline{\anti}  \cong  \dcol \oplus \dhook \oplus \dboxx,
\]
where $\overline{\anti}$ denotes the representation dual to $\anti$. These irreducible representations have dimensions $1$, $d^2-1$ and $\left( \frac{d(d-1)}{2} \right)^2-d^2$, 
respectively, and their projections are
\begin{align*}
  \Psi &= \frac{2d}{d-1}(P_\anti \ox P_\anti)\bigl( \Phi_{AA'} \ox \Phi_{BB'} \bigr) (P_\anti \ox P_\anti)
        = \proj{\Psi}, \text{ for } \ket{\Psi}
        = \frac{1}{\sqrt{{d\choose 2}}}\sum_{i<j} \ket{\psi_{ij}}\ket{\psi_{ij}}, \\
  Q    &= \frac{2d}{d-2}(P_\anti \ox P_\anti)\bigl( (\1-\Phi)_{AA'} \ox \Phi_{BB'} \bigr) (P_\anti \ox P_\anti), \\
  \PP  &= P_\anti \ox P_\anti - Q - \Psi.
\end{align*}
\end{lemma}
\begin{proof}
The abstract decomposition follows from $\overline{\anti}\cong  \dualcol$ and from 
the Littlewood-Richardson rule that governs the decomposition of tensor 
products of irreducible representations of $SU(d)$ (see e.g~\cite{Fulton97}). The dimensions follow from Weyl's formula.

For the explicit form of the projectors, we only need to guess
the invariant one-dimensional subspace, and one other invariant
operator, which are our $\Psi$ and $Q$ -- since they are orthogonal
to each other and have the correct trace, they must be projectors. 
The third one is then their
complement with respect to $P_\anti \ox P_\anti$.
\end{proof}

\begin{lemma}
  \label{lem:pt-overlaps}
  For $\Psi$ and $Q$ as in Lemma~\ref{lem:Psi-Q-PP},
  \[
    \tr\rho_\col^\Gamma \Psi  =  \frac{2}{d(d-1)},\quad
    \tr\rho_\boxx^\Gamma \Psi =  \frac{2}{d(d-1)},\quad
    \tr\rho_\hook^\Gamma \Psi = -\frac{2}{d(d-1)},
  \]
  and
  \[
    \tr\rho_\col^\Gamma Q  = -\frac{2(d+1)}{d(d-2)},\quad
    \tr\rho_\boxx^\Gamma Q =  \frac{1}{d},\quad
    \tr\rho_\hook^\Gamma Q = -\frac{2}{d(d-2)}.
  \]
  (Then the expectations of $\PP$ are determined by 
  $\tr\rho_y^\Gamma\PP = 1-\tr\rho_y^\Gamma\Psi-\tr\rho_y^\Gamma Q$.)
\end{lemma}
\begin{proof}
  For the expectations of $\Psi$, note that
  \[\begin{split}
    \tr\rho_y^\Gamma \Psi &= \frac{2d}{d-1}\tr\rho_y^\Gamma(\Phi_{AA'}\ox\Phi_{BB'}) \\
                          &= \frac{2d}{d-1}\frac{1}{d^2}\tr\rho_y (F_{AA'}\ox F_{BB'})
  \end{split}\]
  since $\Phi^\Gamma = \frac{1}{d}F$. From the symmtries
  of the irreducible representations we know that 
  $\tr\rho_\col (F_{AA'}\ox F_{BB'}) = \tr\rho_\boxx (F_{AA'}\ox F_{BB'}) = 1$
  and $\tr\rho_\hook (F_{AA'}\ox F_{BB'}) = -1$.
  
  For $Q$, we proceed similarly:
  \[\begin{split}
    \tr\rho_y^\Gamma Q &= \frac{2d}{d-2}\tr\rho_y^\Gamma\bigl( (\1-\Phi)_{AA'}\ox\Phi_{BB'}\bigr) \\
         &= \frac{2d}{d-2}\tr\rho_y \left( \left(\1-\frac{1}{d}F_{AA'}\right) \ox \frac{1}{d}F_{BB'} \right) \\ 
         &= \frac{2}{d-2}\tr\tilde\rho_y F_{BB'} - \frac{2}{d(d-2)} \tr\rho_y(F_{AA'}\ox F_{BB'}),
  \end{split}\]
  where we have used the partial traces $\tilde\rho_y = \tr_{AA'} \rho_y$
  from Lemma~\ref{lem:tildes}. The same lemma and the symmetries of the $\rho_y$
  already used above yield the claimed values.
\end{proof}

\section{The Linear Programme}
\label{app:B}
Here we record some observations on the linear programming
relaxation studied in Section~\ref{sec:cost}.

\bigskip\noindent
{\bf The cases of $\mathbf{n=1,2}$, $\mathbf{4}$, \ldots, $\mathbf{12}$.}
For $n=1$ the linear programme is nearly trivial, and indeed it can be seen almost
immediately that the optimal solution is $p_{\col}=0$, $p_{\boxx}=1$, giving a value
of $1/2$ for the objective function.

For $n=2$, the objective function is given by
\[
  \vec{t}^{\,\ox 2} = \left[ 1,\ -\frac{1}{2},\ -\frac{1}{2},\ \frac{1}{4} \right],
\]
while the constraint matrix is
\[
  T^{\ox 2} = \left[\begin{array}{rrrr}
                 1 &  1 &  1 & 1 \\
                -2 &  1 & -2 & 1 \\
                -2 & -2 &  1 & 1 \\
                 4 & -2 & -2 & 1
              \end{array}\right].
\]
From this it becomes clear by inspection of the LP that the optimal vector has
the form $\vec{p} = [ x,\ 0,\ 0, 1-x ]^\top$, leaving as the only
nontrivial constraint, apart from $0\leq x\leq 1$, that
$-2x + (1-x) \geq 0$. Consequently, the optimal solution is $x=1/3$,
yielding a maximum value of $1/2$ of the objective function.
I.e., our method cannot give anything better than 
$E_C(\alpha_d) \geq 0.5$
For $n=4$, one can confirm (using a computer) that the optimal value
is $1/4$; for $n=6$ it is $1/7$, and for $n=8$, $n=10$ and $n=12$,
one finds optimal values
$\frac{5}{66} \approx 0.075757$, 
$\frac{12}{283} \approx 0.0424023$ 
and $\frac{26}{1119} \approx 0.023235$. 
The latter shows that the best lower bound obtainable with the
present method cannot be better than
$E_C(\alpha_d) \geq \frac{1}{12}\log_2\frac{1119}{26} \approx 0.452$.

\end{appendix}


%

\end{document}